\documentclass[letterpaper]{ctistyle}

\usepackage[english]{babel}

\usepackage{cite}
\usepackage{graphicx}
\usepackage[usenames]{color}
\usepackage{subfigure}
\usepackage{float,afterpage}
\usepackage{comment}
\usepackage{amssymb}
\usepackage{subfigure}
\usepackage{ifthen}
\usepackage{url}
\usepackage{paralist}

\normallatexbib

\includecomment{Short}
\excludecomment{Long}
\excludecomment{LongFUTURE}
\excludecomment{LongNOT}

\newcommand{\LONG}[1]{}

\newcommand{\FIGURE}[1]{#1}

\newcommand{\cu}[1]{\mathcal{#1}}
\newcommand{\Real}{{\mathbb{R}}}

\newcommand{\XXX}[1]{}
\newcommand{\remove}[1]{}

\newcommand{\ForgetXXX}[1]{}

\newcommand{\myqed}{\hfill$\Box$}

\graphicspath{{./}{./Graph/}}

\newcommand{\NOBfull}{Increased Rate for Exceptional Nodes,
Identical Rate for Other Nodes}
\newcommand{\NOB}{IREN/IRON}

\newcommand{\Lat}{\mathcal{L}}
\newcommand{\LatInt}{\Lat_i}
\newcommand{\LatBorder}{\Delta{}\Lat}

\newcommand{\LatMap}{\lambda}
\newcommand{\LatRevMap}{\lambda^{-1}}
\newcommand{\DS}{\Delta{}S}
\newcommand{\CLat}{C^{(\Lat)}}
\newcommand{\mmin}{m_\mathrm{min}}
\newcommand{\mmax}{m_\mathrm{max}}
\newcommand{\Cmin}{C_\mathrm{min}}

\newcommand{\Ecost}{E_\mathrm{cost}}
\newcommand{\Erel}{E_\mathrm{rel-cost}}
\newcommand{\Ebound}{E_\mathrm{bound}}
\newcommand{\Mmax}{M_\mathrm{max}}

{}
\newcommand{\yperedge}{yperarc}{}

\def\inprobHIGH{\,{\buildrel p \over \rightarrow}\,}

\newcommand{\mySmath}[1]{$#1$}
\newcommand{\LONGTC}[1]{}

\newcommand{\CORRECT}[1]{#1}

\newtheorem{property}{Property}

\newcommand{\reffig}[1]{Fig.~\ref{#1}}
\newcommand{\refth}[1]{Th.~\ref{#1}}
\newcommand{\reflem}[1]{lemma~\ref{#1}}
\newcommand{\refsec}[1]{section~\ref{#1}}
\newcommand{\refprop}[1]{Property~\ref{#1}}

\newcommand{\refeq}[1]{(\ref{#1})}
\newcommand{\mymath}[1]{$#1$}
\newcommand{\SQUEEZE}[1]{}

\begin{document}

\articletitle{Near Optimal Broadcast\\
~with Network Coding\\
~in Large Sensor Networks}

\author{C\'edric Adjih}%
\affil{Hipercom Team, INRIA Rocquencourt, France}%
\email{Cedric.Adjih@inria.fr}%
\author{Song Yean Cho}%
\affil{Hipercom Team, LIX, \'Ecole Polytechnique, Palaiseau, France}%
\email{Cho@lix.polytechnique.fr}%
\author{Philippe Jacquet}%
\affil{Hipercom Team, INRIA Rocquencourt, France}%
\email{Philippe.Jacquet@inria.fr}%

\begin{abstract}
We study efficient broadcasting for wireless sensor networks,
with network coding.
We address this issue for homogeneous sensor networks in the plane.
Our results are based on a simple principle (IREN/IRON),
which sets the same rate on most of the nodes (wireless links)
of the network. With this rate selection, we give
a value of the maximum achievable broadcast rate of the source:
our central result is a proof of the value of the \emph{min-cut}
for such networks, viewed as \emph{hypergraphs}.
Our metric
for efficiency is the number of transmissions necessary to transmit
one packet from the source to every destination:
we show that IREN/IRON
achieves near optimality for large networks;
that is, asymptotically, nearly every transmission brings new
information from the source to the receiver. As a consequence,
network coding asymptotically outperforms any scheme that does not use
network coding.

\end{abstract}

{}
\section*{Introduction}{}
\label{sec:introduction}

Seminal work{}
in \cite{Bib:ACLY00}
has introduced the idea of {\em network coding},
whereby intermediate nodes are mixing information from different flows
(different bits or different packets)%
{}.%
{}%

One logical domain of application is \emph{wireless sensor networks}.
Indeed, for wireless networks, a generalization of the results in 
\cite{Bib:ACLY00} exists: when 
{}
the capacity of the links are known and fixed, the maximal broadcast
rate of the source can be computed, as shown in
{}\cite{Bib:DGPHE06}{}.
Essentially, for one source, it is the min-cut
of the network{}
from the source to the destinations,
as for wired networks \cite{Bib:ACLY00}, but considering
\emph{hypergraphs} rather than graphs. This is true whether the rate and
the capacity are expressed in bits per second or packets per second
\cite{Bib:LMKE07}.

However, in wireless sensor networks, 
a primary constraint is not necessarily
the capacity of the wireless links: 
because of the limited battery of each node,
the limiting factor is the cost of wireless transmissions.
Hence a different focus is \emph{energy-efficiency}, rather
than the maximum achievable broadcast rate:
\begin{quote}
$\bullet$ Given one source,
minimize the total number of transmissions used to achieve the
broadcast to destination nodes.
\end{quote}

The problem is no longer related to the capacity, because the 
same transmissions can be streched in time, with an identical cost.
However, one can still imagine using network coding, where 
each node repeats combinations of packets with an average
interval between transmissions: this defines the rate of the node,
and the rate is an unknown. 

The problem of energy-efficiency is to compute
a set of transmission rates for each node, with minimal cost.
With network coding, the problem turns out to be solvable in
polynomial time: for the stated problem, 
\cite{Bib:WCK05,Bib:LRMKKHAZ06} describe methods
to find the optimal
transmission rate of each node with a linear program.
However, this does not
necessarily
provide direct insight about the optimal rates and their
associated optimal cost: those are obtained by solving the linear program
on instances of networks.

For large-scale sensor networks, one assumption could be that the nodes are
distributed in a homogeneous way, and a question would be: 
``Is there a simple near-optimal rate selection ?''
Considering the results of min-cut estimates for random graphs
\cite{Bib:RSW05,Bib:AKMK07,Bib:CB07},
one intuition is that most nodes have similar neighborhood; hence
the performance, when setting an identical rate for each node, deserves
to be explored.
This is the starting point of this paper, and
we will focus on homogeneous networks, 
which can be modeled as unit disk graphs:
\begin{compactenum}
\item We introduce a simple rate {}principle where most nodes
  have the same transmission rate: \NOB{} principle (\NOBfull{}).
\item We give a proof for the min-cut for some lattice graphs (modeled
  as hypergraphs). It is also an intermediate step for the following:
\item We deduce an  estimate of the min-cut for unit disk hypergraphs.
\item We show that this simple rate selection achieves ``near
 optimal performance'' in some classes of homogeneous networks, based on
min-cut computation --- and may outperforms any scheme
that is not using network
coding.%
{}
\end{compactenum}
The rest of this paper is organized as follows: \refsec{sec:network-model}
details the network model and related work; \refsec{sec:iren-iron}
describes the main results{};
\refsec{sec:min-cut-proof} gives proofs of the min-cut; and 
\refsec{sec:conclusion} concludes.

\section{Network Model and Related Work}{}%
{}
\label{sec:network-model}

In this article, we study the problem of broadcasting from one source
to all nodes.
We will assume an ideal wireless model,
wireless transmissions
without loss, collisions or interferences and that each node of
the network is operating well below its maximum transmission capacity.

Our focus is on large-scale wireless sensor networks.
Such networks have been modeled as
\emph{unit disk graphs} \cite{Bib:CCJ02} of the plane, 
where two nodes are neighbors whenever their distance is lower than
a fixed radio range; see \reffig{fig:unit-disk} for the principle
of unit disk graphs.
\FIGURE{%
\begin{figure}[htp]
\centering
\subfigure[Unit disk graph {}]{
\label{fig:unit-disk}
\includegraphics[width=.22\textwidth]{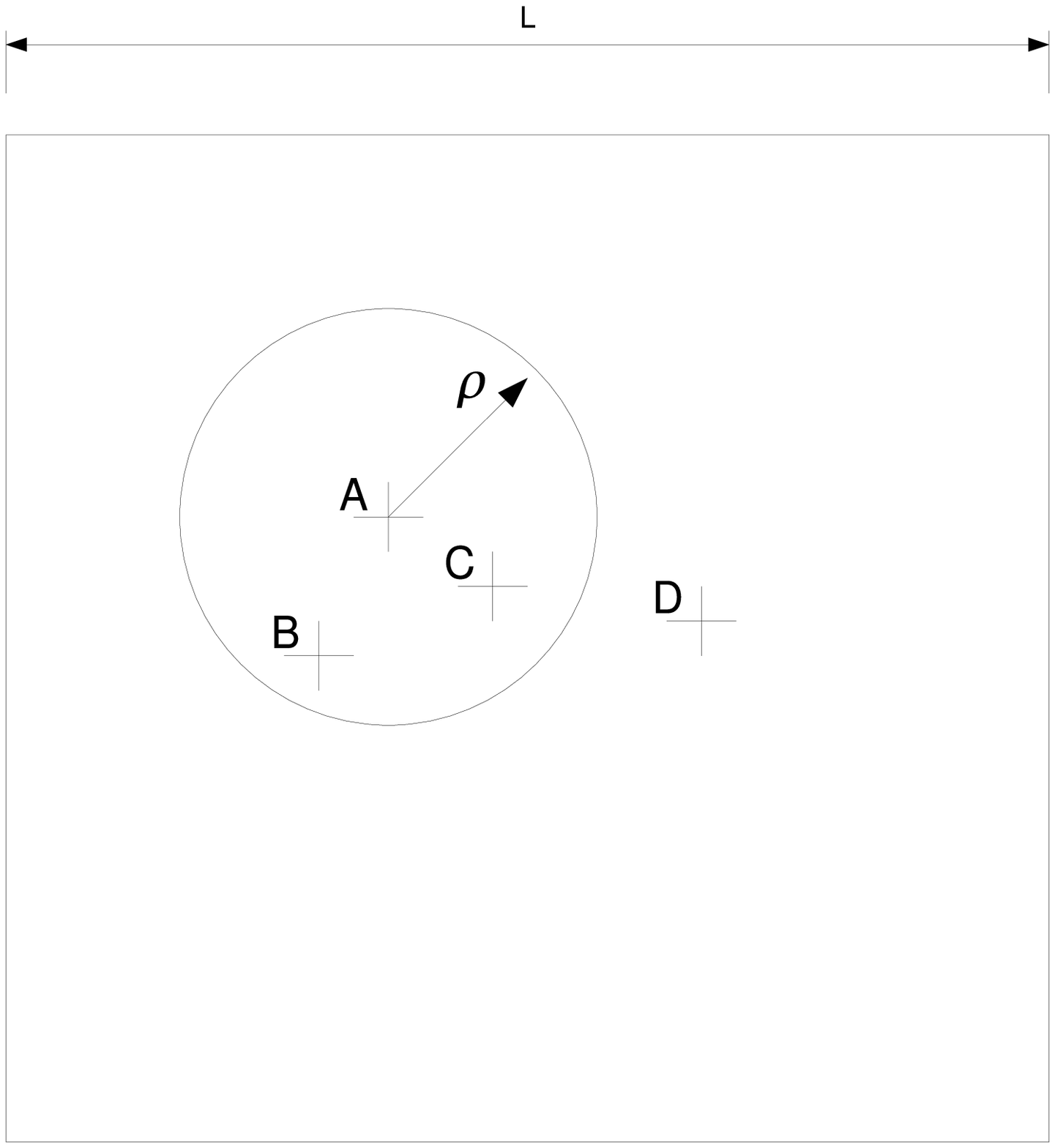}
}%
\hspace{.1in}%
\subfigure[Lattice]{
\label{fig:lattice}
\includegraphics[width=.22\textwidth]{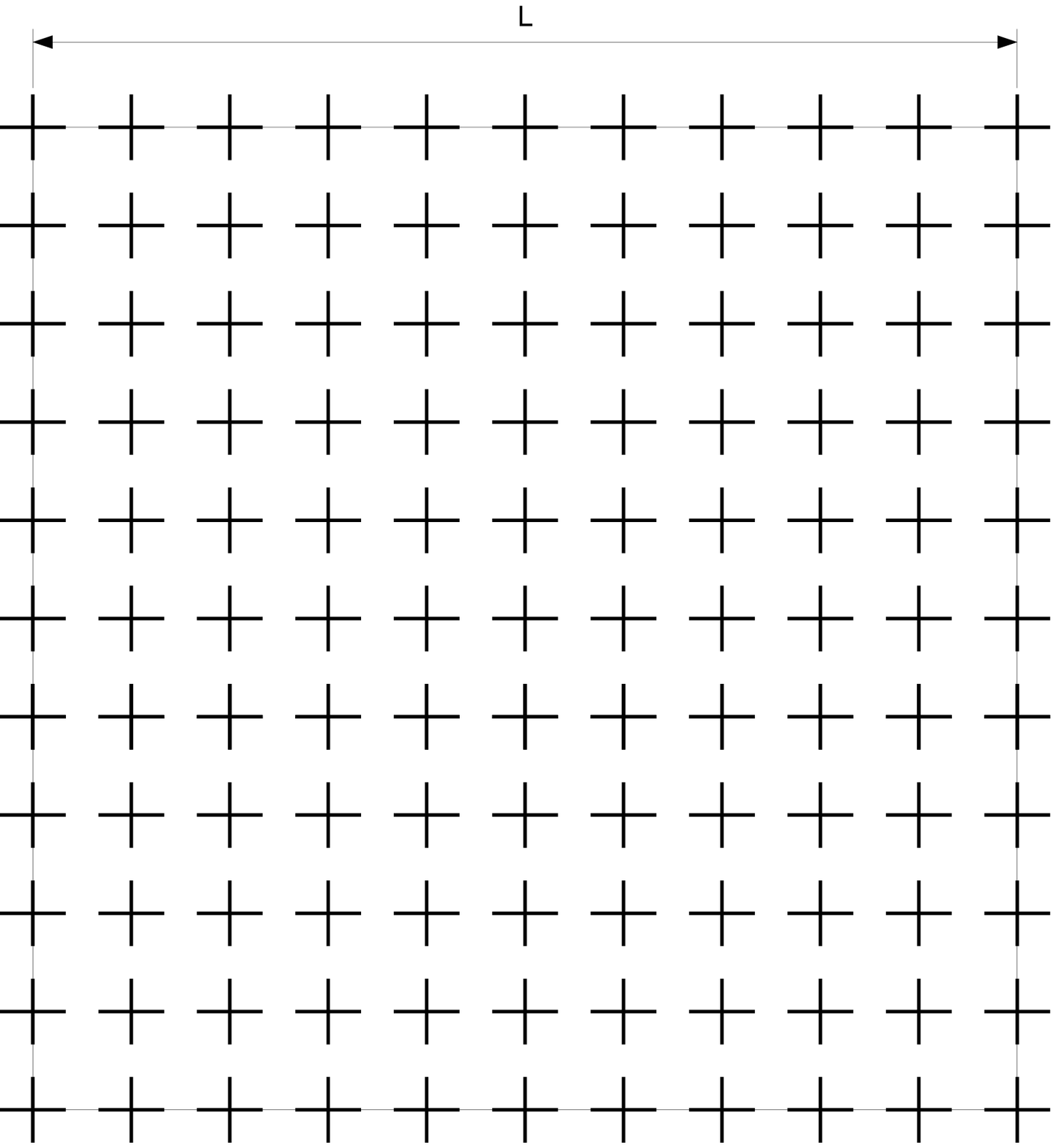}
}%
\hspace{0.01in}
\subfigure[Unit disk range $R$ for a lattice]{
\label{fig:lattice-disk}
\includegraphics[width=.23\textwidth]{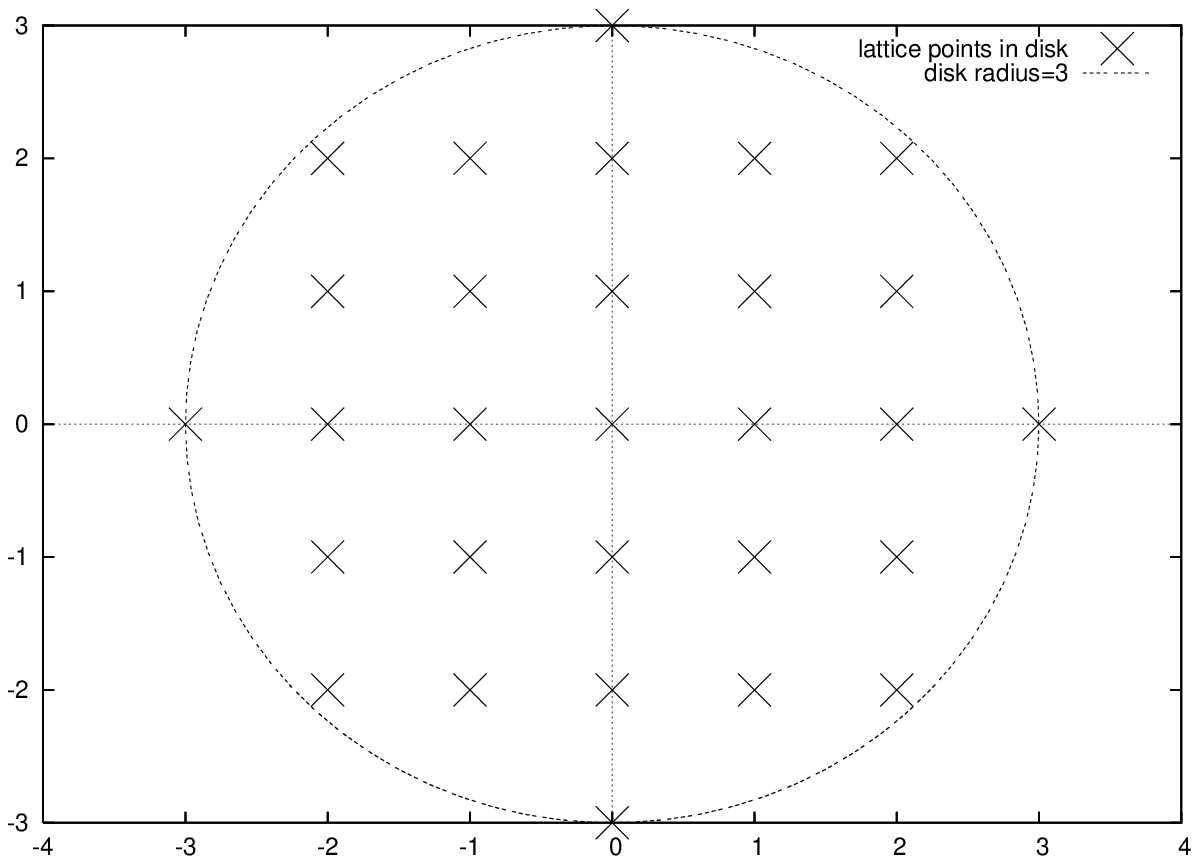}
}
{}%
\vspace{-3mm}
\caption{Network Models}
\end{figure}
}%
\vspace{-3mm}%
Precisely, the sensor networks considered will be:
\begin{compactitem}
\item Random unit disk graphs with nodes uniformly distributed
(\reffig{fig:unit-disk})
\item Unit disk graphs with nodes organized on a lattice
(\reffig{fig:lattice}).
\end{compactitem}
An important assumption is that the \emph{wireless broadcast advantage}
is used: each transmission is overheard by several nodes. As a result
the graph is in reality a \emph{(unit disk) hypergraph}%
{}%
{}.

\subsection{Related Work}

In general, specifying the network coding protocol is reduced to specifying
the transmission rates for each node \cite{Bib:LMKE05}. %
{}%
Once the optimal rates are
computed, the performance can be asymptotically achieved with 
\emph{distributed random linear coding},
for instance \cite{Bib:HKMKE03,Bib:LMKE07}.
{}This article is in the spirit of \cite{Bib:FWB06}, which
starts with exhibiting an energy-efficient algorithm for simple networks.
The central element for computing the performance is the estimation of
the min-cut of the network and we are inspired by the existing
techniques and results surrouding the expected value of the min-cut 
on some classes
of networks: for instance, \cite{Bib:RSW05} explored the capacity
of networks
where a source is two hops from the destination, through a
network of relay nodes; \cite{Bib:AKMK07} studied some classes
of random geometric graphs.
Recently, \cite{Bib:CB07} gave bounds of the min-cut of dual radio networks.

{}
\subsection{Network Coding}{}
\label{sec:network-coding}

In the network coding literature, several results are for multicast,
and, in this section, they are quoted as such. They
apply to the topic of this article, broadcast, since broadcast is a 
special case of multicast. 

A central result for %
network coding in wireless networks 
gives
the maximum multicast rate for a source.
The capacity 
is given by the min-cut
from the source to each individual destination of the network,
viewed as a hypergraph
\cite{Bib:DGPHE06,Bib:LRMKKHAZ06}.%
\label{sec:hypergraph-min-cut}%
 A precise description includes:\\{}%
{}%
{}%
{\bf $\bullet$ Nodes:} $\cu{V} = \{ v_i, i=1,\ldots{}N \}$, set of 
nodes of the hypergraph {}\\
{\bf $\bullet$ H\yperedge{}:}  $h_v = (v, H_v)$, where $H_v \subset V$ is
  the subset of nodes that are reached by one transmission of node $v$%
 (neighbors){}{}.\\
{}
{\bf $\bullet$ Rate:} Each node $v$ emits on the h\yperedge{} $(v, H_v)$
  with {}rate $C_v$.
{}%

Let us consider the source $s$, and one of the multicast
destinations $t \in \cu{V}$. {} 
The definition of an \emph{$s$-$t$ cut} is: a partition of the set
of nodes $V$ in two sets $S$, $T$ such as $s \in S$ and $t \in T$.
Let $Q(s,t)$ be the set of such \emph{$s$-$t$ cuts}: $(S,T) \in Q(s,t)$.
{}%

{}%

We denote $\DS$, the set of nodes of $S$ that are neighbors of at
least one node of $T$;
the \emph{capacity of the cut} $C(S)$ is defined as the maximum rate
between the nodes in $S$ and the nodes in $T$:
\vspace{-3mm}
\begin{equation}
\label{eq:deltaS}
\label{eq:cut-capacity}
\DS \triangleq \{ v \in S : H_v \cap T \neq \emptyset \}
~~~\mathrm{and}~~~
C(S) \triangleq \sum_{v \in \Delta{}S} C_v
\end{equation}
\vspace{-3mm}%
{}%

{}

The \emph{min-cut} between $s$ and $t$ is the cut of $Q(s,t)$ with the
minimum capacity. Let us denote $C_\mathrm{min}(s,t)$ as its capacity.
From \cite{Bib:DGPHE06,Bib:LRMKKHAZ06},
the maximum multicast capacity is given by the minimum of capacity
of the min-cut of every destination, $C_\mathrm{min}(s)$, with:
\vspace{-3mm}
\begin{equation}
\label{eq:s-t-cut}
C_\mathrm{min}(s,t) \triangleq \min_{(S,T) \in Q(s,t)} C(S)
~~~\mathrm{and}~~~
C_\mathrm{min}(s) \triangleq \min_{t \in \cu{V} \setminus \{s\}} 
C_\mathrm{min}(s,t)
\end{equation}
\vspace{-6mm}%
{}%

{}

{}
\section{Main Results}{}%
\label{sec:iren-iron}

\subsection{Overview}

As described in the introduction, our approach is to choose
an intuitive transmission rate for each node:
essentially, the same rate for most nodes%
{},{} as described 
in \refsec{sec:rate}.
Then, we determine the maximum broadcast rate that can be achieved to transmit
from the source to every
node in the network as the min-cut of the hypergraph,
for both random and lattice graphs
in \refsec{sec:min-cut-overview}. And finally,
from the expression of the cost in~\refsec{sec:cost}, we deduce asymptotic
optimality (\refsec{sec:near-optimal}).

\subsection{Further Definitions}
\label{sec:futher-definitions}
{}%

Consider a network inside a square area $G$ of edge length $L$,
such as the one on \reffig{fig:unit-disk}.
{}%
\begin{compactitem}
\item The radio range of the network is $\rho$.
\item For a lattice, we denote $R$ the set of neighbors of the origin
node $(0,0)$, as represented on \reffig{fig:lattice-disk}:
$R \triangleq \{ (x,y) \in \mathbb{Z}^2 : x^2+y^2 \le \rho^2 \}$
\item Let $M$ be the ``expected'' number of neighbors of one node.
For a lattice, it is $M = |R|-1$.
For a random disk unit graph with $N$ nodes,
$M$ is related to the density $\mu = \frac{N}{L^2}$ and range as follows:
$M=\pi \rho^2 \mu = \pi \rho^2 \frac{N}{L^2}${}.
\end{compactitem}
We define the \emph{border area} as the area of fixed width $W > \rho$ 
near the edge of that square, and \emph{border nodes}
as the nodes {}in that area. Hence,
the area $L \times L$ of $G$ is partitioned into:
\begin{compactitem}
\item $\Delta{}G$, the border, with area $A_{\Delta{}G} = 4 W (L-W)$
{}
\item $G_i$, the ``interior'' $G_i \triangleq G \setminus \Delta{}G$, 
with area $A_{G_i} = (L - 2 W)^2$
\end{compactitem}

\subsection{Rate Selection with \NOB}
\label{sec:rate}

The principle \NOB{} sets the following transmission rates:
{}%
\begin{compactitem}{}
\item IREN (Increased Rate for Exceptional Nodes): 
  the rate of transmission is set to $M$, for the source node
  and all the border nodes (the ``exceptional'' nodes).
\item IRON (Identical Rate for Other Nodes):
  every other node{} transmits with rate $1$.
\end{compactitem}{}
{}%

{}

\subsection{Performance: Min-Cut (Achievable Broadcast Rate)}
\label{sec:min-cut-overview}

The essence of our main result is the following property
proved in \refsec{sec:lattice-proof}, \refth{th:min-cut-lattice}:

\begin{property}
\label{prop:lattice}
With the rate selection \NOB{}, the min-cut of a lattice graph 
is {}equal to $C_\mathrm{min} = M$ (with $M = |R|-1$).
\end{property}

For random unit disk graphs, 
by mapping the points to an imaginary lattice graph
(\emph{embedded lattice}) as an intermediary step,
we are able to find bounds of the capacity of random unit disk
graphs. This turns out to be much in the spirit 
of \cite{Bib:CB07}.
This is used to
deduce an asymptotic result for unit disk graphs,
proven in \refsec{sec:unit-disk-proof}, \refth{th:Cmin-limit}:
\begin{property}
\label{prop:unit-graph}
Assume a fixed range.
For a sequence of random unit disk graphs $(\cu{V}_i)$,
with sources $s_i$,
with size $L \rightarrow \infty$ and with a density
$M \rightarrow \infty$ such as $M = {}L^\theta{}$,
for any fixed $\theta>0$, we have the following convergence in probability:
$\frac{C_{\min}(s)}{M} \inprobHIGH 1$.

\end{property}

\subsection{Performance: Transmission Cost per Broadcast}
\label{sec:cost}

Recall that the metric for cost
is the number of (packet) transmissions per a (packet) broadcast from
the source to the entire network. 
Let us denote $E_\mathrm{cost}$ as this ``transmissions per broadcast.''

This cost of broadcasting with \NOB{} rate selection
can be equivalently computed from the rates as the ratio
of the number of transmissions per unit time to the number of packets
broadcast into the network per unit time.
 {}
Then $\Ecost$ is deduced from the min-cut $\Cmin$,
{} the areas $A_{\Delta{}G}, A_{G_i}$,
the associated node rates
{} and the node density $\mu$.
For fixed $W$, $M,L \rightarrow \infty$:
\mymath{E_\mathrm{cost} = \frac{1}{\Cmin} \mu L^2 \left(1 + O(\frac{1}{L}) 
+ \frac{4MW}{L} (1+ O(\frac{1}{L})) \right)}.{}

For random unit disk graphs{}, $E_\mathrm{cost}$ is an
expected value{}, and
$\mu=\frac{N}{L^2}$. For a lattice, $\mu = 1$.
\subsection{Near Optimal Performance for Large Networks}
\label{sec:near-optimal}

Sections~\ref{sec:min-cut-overview} and~\ref{sec:cost} gave the
performance and cost with the {}\\IREN/IRON{} principle.
{}The optimal cost
is not easily computed, and in this section an
indirect route is chosen, by using a bound.

Assume that every node has at most $M_\mathrm{max}$ neighbors: 
one single transmission can provide information to
$M_\mathrm{max}$ nodes at most. Hence, in order
to broadcast one packet to all $N$ nodes, at least 
$\Ebound = \frac{N}{M_\mathrm{max}}$
transmissions are necessary. 

{}
W.r.t. this bound, let the relative cost be:
$\Erel = \frac{\Ecost}{\Ebound} \ge 1$.

We will prove that $\Erel \rightarrow 1$
for the following networks:

\subsubsection{Lattice Graphs{}}

{}
For lattice{}s, $W$ and the neighborhood $R$ are kept fixed 
(hence also $M = |R|-1$)%
{} and
only the size $L$ of the network increases to infinity. The number
of nodes is $N=L^2${}.
The maximum number of neighbors $M_\mathrm{max}$ is exactly 
$M_\mathrm{max} = M$. 

From \refsec{sec:cost} and from \refprop{prop:lattice}, we have:\\
\mymath{E_\mathrm{rel-cost}=E_\mathrm{cost}\frac{M_\mathrm{max}}{N}
= \left(1 + O(\frac{1}{L}) 
+ \frac{4MW}{L} (1+ O(\frac{1}{L})) \right) = 1+O(\frac{1}{L})}{}%
{}.

\subsubsection{Random Unit Disk Graphs{}}

For random unit disk graphs, first notice that an increase of the density 
$M$ does not improve the relative cost 
$\Erel${}.
Now consider a sequence of random graphs, as in \refprop{prop:unit-graph},
with fixed {}$\rho$, fixed {}$W$, and
size $L \rightarrow \infty$ and with a density
$M \rightarrow \infty$ such as $M = {}L^\theta{}$,
for some arbitrary fixed $\theta>0$, with the additional constraint that 
$\theta < 1$. We have:

\mymath{\Erel = \Ecost\frac{\Mmax}{N}
= \frac{M}{\Cmin}\frac{\Mmax}{M}\frac{\mu L^2}{N}\left(1 + O(\frac{1}{L}) 
+ \frac{4MW}{L} (1+ O(\frac{1}{L})) \right)}.{}%
{}%
\CORRECT{

Each of part of the product converges toward $1$, either surely, 
or in probability:
using \refprop{prop:unit-graph},
we have the convergence of $\frac{\Cmin}{M} \inprobHIGH 1$,
when $L \rightarrow \infty$ and similarly with \refth{th:Cmin-limit} 
we have $\frac{\Mmax}{M} \inprobHIGH 1$. 
By definition $N = \mu L^2$. Finally, 
$M = {}L^\theta{}$ for $\theta < 1$ implies that 
$\frac{4MW}{L} \rightarrow 0$.

As a result we have:
\mymath{\Erel \inprobHIGH 1} in probability, when $L \rightarrow \infty$}

\subsubsection{Random Unit Disk Graphs without Network Coding{}}

In order to compare the results that are obtained when network coding
is not used, one can reuse the argument of \cite{Bib:FWB06}.
Consider the broadcasting of one packet.
Consider one node of the network that has repeated the packets.
It must have received the transmission
from another connected neighbor.
In a unit disk graph, these
two connected neighbors share a neighborhood area at least equal to
$(\frac{2 \pi}{3} - \frac{\sqrt{3}}{2}) \rho^2$, and every node
lying within that area will receive duplicate of the packets. 
Considering this inefficiency,
for dense unit disk graphs one can deduce the following bound:
\mymath{\Erel^\mathrm{(no-coding)} \ge \frac{6 \pi}{2 \pi + 3 \sqrt{3}}}
. Notice that $\frac{6 \pi}{2 \pi + 3 \sqrt{3}} \approx 1.6420\ldots > 1$.

\subsubsection{Near Optimality{}}
\label{sec:near-optimality}

The asymptotic optimality is a consequence of the convergence of the
cost bound $\Erel$ toward $1$. 
{}This{} indirect proof is in fact 
a stronger statement than optimality of the rate selection
in terms of energy-efficiency: it exhibits the fact that asymptotically
(nearly) all the transmissions will be \emph{innovative} for the receivers.
Note that it is not the case in general for a given instance of a hypergraph.
It evidences the following remarkable fact for the large homogeneous
networks considered: 
network coding may be achieving 
not only optimal efficiency, but
also, asymptotically, perfect efficiency ---
achieving the information-theoretic bound for each transmission.

Notice that traditional broadcast methods without network coding 
(such as the ones based on connected dominating sets) cannot
achieve this efficiency, since their
lower bound is $1.642$.

\section{Proofs of the Min-Cut}{}
{}
\label{sec:min-cut-proof}

\subsection{Proof for Lattice Graphs}
\label{sec:lattice-proof}
\subsubsection{Preliminaries{}}
Let $\Gamma$ be full, {}\emph{integer lattice} 
in $n$-dimensional space;
it is the set $\mathbb{Z}^n$, where the lattice points are $n$-tuples
of integers.

For lattice graphs, only points on the full lattice are relevant; therefore in
this section, the notations $\Lat, \LatInt, \LatBorder$ will be used,
for the parts of the full lattice $\Gamma$ that are in $G, G_i, \Delta G$
respectively.{}

The proof is based on the use of the Minkowski addition, and 
a specific property of discrete geometry~\refeq{eq:mink-ineq} below.
The Minkowski addition is a classical way to express
the neighborhood of one area (for instance, see \cite{Bib:LKE98} and 
the figure 3(a), and figure 4 of that reference).

Given two sets $A$ and $B$ of $\Real^n$, the Minkowski sum 
of the two sets  $A \oplus B$ is defined as the set of all
vector sums generated by all pairs of points in $A$ and $B$,
respectively:
\mymath{A \oplus B \triangleq \{ a + b : a \in A, b \in B \}}

{}
Then the set of neighbors $\cu{N}(t)$ of one node $t$, with $t$ itself, is:\\
\mymath{\cu{N}(t) \cup \{ t \} = \{ t \} \oplus R}
\\
This extends to the neighborhood of a set of points.
For Minkowski sums on the lattice $\Gamma$, there exist variants of the 
\emph{Brunn-Minkowski inequality}, including the following
one \cite{Bib:GG01}:
\begin{property}
For two subsets $A, B$ of the integer lattice $\mathbb{Z}^n$,
\begin{equation}
\label{eq:mink-ineq}
|A \oplus B| \ge |A| + |B| - 1
\end{equation}
where $|X|$ represents the number of elements of a subset $X$
of $\mathbb{Z}^n$
\end{property}

\subsubsection{Bound on the capacity of one cut $C(S)${}}
{}
Consider a lattice $\Lat$ and a source $s$.
{}
Let $C(S)$ be the capacity of {} an $s$-$t$ cut $S,T \in Q(s,t)$.

{}
\begin{lemma}
\label{lem:cut-bound}
$C(S) \ge |\DS|$ (with $\DS$ defined in
{}\refeq{eq:cut-capacity}{})
\end{lemma}
\begin{proof}
$C_v \ge 1$ with \NOB{} and with \refeq{eq:cut-capacity},
$C(S) = \sum_{v \in \DS} C_v$
\end{proof}
{}

\begin{theorem}
\label{th:capacity-bound}
The capacity of one cut $C(S)$ is such that: 
\mymath{C(S) \ge M}
\end{theorem}
\emph{Proof}:
There are three possible cases, 
either the set $T$ has no common nodes 
with the border $\LatBorder$, or $T$ includes all nodes of $\LatBorder$, 
or finally $T$ includes only part of nodes in the border area. 

{\bf First case}, $T \cap \LatBorder = \emptyset$:

{}We{}
know that $T \oplus R \subset \Lat$,
hence we can effectively write the neighbors of nodes in $T$ as a Minkowski
addition (without getting points in $\Gamma$ but out of $\Lat$):
\mymath{\Delta T \triangleq (T \oplus R) \setminus T}

It follows that:
\mySmath{|\Delta{}T| \ge |T \oplus R | - | T|}

Now the inequality~(\ref{eq:mink-ineq}) can be used:
\mySmath{|T \oplus R| \ge |T| + |R| -1}

Hence we get:
\mySmath{|\Delta T| \ge |T| + |R| - 1 - |T|}, and therefore: 
\vspace{-1mm}
\begin{equation}
\label{eq:tmp1}
|\Delta T| \ge |R| - 1
\end{equation}
\vspace{-6mm}

Recall that $S$ and $T$ form a partition of $\Lat$ ;
and since $\Delta T$ is a subset of $\Lat$, by definition without any point of
of $T$, we have $\Delta T \subset S$. 
Hence actually $\Delta T \subset \Delta S$ (with the definition of 
$\Delta S$ in \refeq{eq:cut-capacity}).
We can combine this fact with \reflem{lem:cut-bound}
{}%
and \refeq{eq:tmp1}, to get:

$|C(S)| \ge |R| - 1$
and the \refth{th:capacity-bound} is proved for the first
case.

The second case is similar, but considering the source, 
while the third case uses the fact that a path can be found in the border
between any two border nodes
\cite{Bib:preprint}. \myqed{}%

\subsubsection{Value of the Min-cut $\Cmin(s)$}\hfill{}
\label{sec:th-lattice}
The results of the previous section immediately result in
a property on the capacity of every $s$-$t$ min-cut:
\begin{theorem}
\label{th:min-cut-lattice}
For any $t \in \Lat$ different from the source $s$:

\mymath{\Cmin(s,t) = M} ; and as a result: $\Cmin(s) = M$

\end{theorem}
\begin{proof}
{}
From~\refth{th:capacity-bound}, we have the capacity of every $s-t$ cut
$S/T$ verifies: $C(S) \ge M$. Hence $\Cmin(s,t) \ge M${}%

Conversely let us consider a specific cut, $S_s = \{ s \}$ 
and $T_s = \Lat \setminus \{ s \}$. Obviously
$s$ has at least one neighbor{} hence $\DS = \{ s \}$.
The capacity of the
cut is $C(S_s) = \sum_{v \in \Delta{}S} C_v = C_s = M$
and thus
$\Cmin(s,t) \le M$, and the theorem follows.\myqed
\end{proof}

\subsection{Proof of the Value of Min-Cut for Unit Disk Graphs}
\label{sec:unit-disk-proof}

In this section, we will prove a probabilistic result on
the min-cut, in the case of random unit disk graphs, using an virtual 
``embedded'' lattice.
The unit graph will be denoted $\cu{V}$, whereas for the embedded
lattice the notation of section~\ref{sec:min-cut-proof} is used:
$\Lat$ (along with $\LatBorder$ and $\LatInt$).
{}

{}

\subsubsection{Embedded Lattice}

Given the square area $L\times L$, we start with fitting a
\emph{rescaled} lattice inside it, with a scaling factor $r$.
Precisely, it is
the intersection of square $G$  and the set $\{ (rx,ry) : (x,y) \in \mathbb{Z}^2 \}$.

We will map the points of $G$ to the closest
point of the  rescaled lattice $\Lat$:
let us denote $\LatMap(x)$, the application that transforms
a point $u$ of the Euclidian space $\Real^2$ 
to its closest point of $\Lat$. Formally, for 
$u = (x,y) \in \mathbb{Z}^2$, 
\mymath{
\LatMap(x) \triangleq (r \lfloor \frac{x}{r} + \frac{1}{2} \rfloor, 
r \lfloor \frac{y}{r} + \frac{1}{2} \rfloor)}

For $u \in \Lat$, $\LatRevMap(u)$
is the set of nodes of $\cu{V}$ that are mapped to $u$.
The area of $\Real^2$ that is mapped to a same point of the lattice,
 is a square
$r \times r$ around that point. 
{}

Let $u$ be a point of the lattice $\Lat$, and let denote
the $m(u)$ the number of points of $\cu{V}$ that are
mapped to $u$ with $g$ (they are in the square around $u$ ; and
$m(u) \triangleq |\LatRevMap(u)|$).
{} $m(u)$ is a random variable.

Let us denote: \mymath{\mmin \triangleq \min_{u \in \Lat} m(u)
\mathrm{~and~} \mmax \triangleq \max_{u \in \Lat} m(u)}

\subsubsection{Neighborhood of the Embedded Lattice}

We start by defining the neighborhood $R$ for the embedded lattice.
{}We{} choose $R(r)$ to be the points of the lattice inside a
disk of radius $\rho - 2r$%
{}%
.{}%

{}
\begin{lemma}
\label{lem:induced-distance}
Let us consider two nodes of $u$, $v$ of $\cu{V}$ that are mapped 
on the lattice $\Lat$ to $u_\Lat$ and $v_\Lat$ respectively: \\
$\bullet$ if $u_\Lat$ and $v_\Lat$ are neighbors on the lattice, them
$u$ and $v$ are neighbors on the graph $\cu{V}$
\end{lemma}
This results from triangle inequalities on the distances.{}%

\subsubsection{Relationship between the Capacities of the Cuts of the Embedded Lattice and the Random Disk Unit Graph{}}
\label{sec:capacity-embedded}
{}

Let us consider one source $s \in \cu{V}$, one
destination $t \in \cu{V}$ and the capacity of any $S/T$ cut.
Every node of $S$ and $T$ is then mapped to the nearest point of the
embedded lattice. For the source, we denote: $s_\Lat = \lambda(s)$.

An \emph{induced cut} of the embedded lattice is constructed as follows:
{}%
\begin{compactitem}{}%
\item The border area width $W_\Lat$ is selected so as to be 
the greatest integer multiple of $r$ which is smaller than 
$W$ ; and $r < W - \rho${}
\item For any point of the lattice $v_\Lat \in \Lat$, 
   the rate $C^{(\Lat)}_{v_\Lat}$ is set according to \NOB{} on the lattice:
   $C^{(\Lat)}_{v_\Lat} = |R(r)|-1 $ when $v_\Lat$ is within the border
  area of width $W_\Lat$, and $C^{(\Lat)}_{v_\Lat} = 1$ otherwise.
\item $S_\Lat$ is the set including the point $s_\Lat$,
and the points of the lattice $\Lat$
such as only nodes of $S$ are mapped to them:
{}%
\\
\mymath{S_\Lat \triangleq \{ s_\Lat \} \cup 
  \{ u_\Lat : \LatRevMap(u_\Lat) \subset S \}}{}
\item $T_\Lat$ is the set of the rest of points of $\Lat$.
\end{compactitem}{}%
{}%
Note that $t \in T_\Lat$ ; that all the points of the lattice,
to which both points from  $S$ and $T$ are mapped, 
those points are in $T_\Lat$ ; and that the points to which no
points are mapped are in $S_\Lat$:
$S_\Lat/T_\Lat$ is indeed a partition and a $s_\Lat-t_\Lat$ cut.

\begin{lemma}
\label{lem:unit-disk-induced-cut-bound}
The capacity $C(S)$ of the cut $S/T$ and the capacity of the induced cut
$\CLat(S_\Lat)$ verify:
\mymath{C(S) \ge m_\mathrm{min} \CLat(S_\Lat)}
\end{lemma}
This comes from 
the fact that neighborhood on the lattice implies
neighborhood in $\cu{V}$ (\reflem{lem:induced-distance}), 
and then an inclusion is proved between the $\DS_\Lat$ 
of the capacity of cut of the lattice from \refeq{eq:s-t-cut}
and the $\DS$ of the cut $S/T$
\cite{Bib:preprint}.{}%

\begin{theorem}
\label{th:unit-disk-min-cut}
The min-cut $\Cmin(s)$ of the graph $\cu{V}$, verifies:

\mymath{\Cmin(s) \ge \mmin (|R(r)|-1)}
\end{theorem}
\begin{proof}

From \reflem{lem:unit-disk-induced-cut-bound},
any cut $C(S)$ is lower bounded by $\mmin \CLat(S_\Lat)$.
Since $\CLat(S_\Lat)$ is the capacity of a cut of a lattice with \NOB{}, 
\refth{th:min-cut-lattice} also indicates that:
$\CLat(S_\Lat) \ge \CLat_\mathrm{min}(s_\Lat) = |R(r)|-1$.
Hence the lower bound 
$\mmin (|R(r)|-1)$ for any $C(S)$, and therefore for the min-cut $\Cmin(s)$.
\end{proof}

\subsubsection{Nodes of $\cu{V}$ Mapped to One Lattice Point{}}

In~\refth{th:unit-disk-min-cut}, $\mmin$ plays a central part.
Let us start with $m(u_\Lat)$: it is actually a random variable
that is the sum of $N$ Bernoulli trials. With a Chernoff
tail bound~\cite{Bib:BL98}, we get, for $\delta \in ]0,1[$:\\
\mymath{Pr[m(u_\Lat) \le (1-\delta) E[m(u_\Lat)]] \le \exp( - \frac{E[m(u_\Lat)] \delta^2}{2} )}

A bound on $\mmin$ is deduced from the fact that it is the minimum of
$m(u)$ and from the fact that for
two events $A$ and $B$, \\
$Pr[ A$ or $B] \le Pr[A] + Pr[B]$:\\
\begin{theorem}
\label{th:bound}
\mymath{Pr[m_\mathrm{min} \le (1 - \delta) \mu r^2] \le \exp 
\left( (\log\frac{L^2}{r^2}) (1 -\frac{\mu r^2\delta^2}{2 \log\frac{L^2}{r^2}}) \right)}
\end{theorem}

\subsubsection{Asymptotic Values of the Min-Cut of Unit-Disk Graphs{}}
\begin{theorem}
\label{th:Cmin-limit}
For a sequence of random unit disk graphs and associated sources
$(\cu{V}_i, s_i \in \cu{V}_i)$, with fixed radio range $\rho$,
fixed border area width $W$, with a size $L_i \rightarrow \infty$,
and a density $M = L^\theta$ with fixed $\theta > 0$, we have
the following limit of the min-cut $\Cmin(s_i)$:\\
\mymath{\frac{\Cmin (s_i)}{M} \inprobHIGH 1 \mathrm{~in~probability. ~
Additionally:~} \frac{\Mmax}{M} \inprobHIGH 1}
\end{theorem}
\begin{proof}
Notice that $Pr[ \mmin > \mu r^2 ] = 0$, because
otherwise we would have a minimum of some values
$m(u_\Lat)$ greater than their average. 
Starting from \refth{th:bound}, several variables appear:
$L$, $\mu$, $\delta$, and $r$. Assume that $\rho$ is fixed
and that $\mu = K L^\theta$ for some fixed $\theta>0$ and $K>0$.
Then we propose the following settings:
$\delta = L^{-\frac{\theta}{8}}$ ;  $r = L^{-\frac{\theta}{8}}$

As a result, from~\refth{th:bound} we have:
$\frac{\mmin}{\mu r^2} \inprobHIGH 1$. Now
\refth{th:unit-disk-min-cut} gives:\\
$\Cmin(s) \ge \mmin (|R(r)|-1)$. Hence:
$\frac{\Cmin(s)}{M} \ge \frac{\mmin}{\mu r^2} \frac{\mu}{M} r^2 (|R(r)|-1) $ \\
From 
the fact that 
$|R(r)| = \pi (\frac{\rho}{r})^2 + O(\frac{1}{r})$ \cite{Bib:preprint},
and that $M=\pi \rho^2 \mu$,  we get the lower bound $1$ for the limit of
$\frac{\Cmin}{M}$. The upper bound comes indirectly from 
the fact that $\frac{\Mmax}{M} \inprobHIGH 1$ 
\cite{Bib:preprint}.\myqed{}
\end{proof}

\section{Conclusion}
\label{sec:conclusion}

We have presented a simple rate selection for network coding for large
sensor networks. We computed the broadcast performance from
the min-cut with networks modeled as hypergraphs.
The central result is that selecting nearly the same rate for all nodes
achieves asymptotic optimality for the homogeneous networks that are
presented, when the size of the networks becomes larger. This can 
be translated into this remarkable property: 
nearly every transmission becomes innovative for the receivers.
As a result, it was shown that network coding would asymptotically
outperform any method that does not use network coding.
We believe that the results presented here are a first step for
a simple but efficient rate selection in wireless sensor networks
in the plane. Future research work will determine
how to adapt the rate selection
for smaller and less homogeneous networks.

{}
{}

\end{document}